\documentclass[letterpaper, 10 pt, conference]{ieeeconf}  

\IEEEoverridecommandlockouts                              

\usepackage{graphicx}
\graphicspath{ {./images/} }
\usepackage{amsmath} 
\usepackage{amssymb}  
\usepackage{mathrsfs,color}

\usepackage{algorithm, algorithmic, xfrac}

\usepackage{theorem}
\newtheorem{lemma}{Lemma}
\newtheorem{definition}{Definition}
\newtheorem{remark}{Remark}

 \newtheorem{theorem}{Theorem}[section]

\newcommand{\enve}{\mathcal{E}}
\DeclareMathOperator*{\argmax}{arg\,max}

\title{\LARGE \bf
Competitive Perimeter Defense in Tree Environments
}

\author{Richard L. Frost$^{1}$ and Shaunak D. Bopardikar$^{2}$
\thanks{This work was sponsored by the Army Research Laboratory and was accomplished under Cooperative Agreement Number W911NF-17-2-0181.}
\thanks{$^{1}$Richard L. Frost is with the Department of Computer Science and Engineering, Michigan State University {\tt\small frostric@msu.edu}}%
\thanks{$^{2}$Shaunak D. Bopardikar is with the Department of Electrical and Computer Engineering, Michigan State University
}%
}

\begin{document}

\maketitle
\thispagestyle{empty}
\pagestyle{empty}

\begin{abstract}

We consider a perimeter defense problem in a rooted full tree graph environment in which a single defending vehicle seeks to defend a set of specified vertices, termed as the perimeter, from mobile intruders that enter the environment through the tree's leaves. We adopt the technique of competitive analysis to characterize the performance of online algorithms for the defending vehicle. We first derive fundamental limits on the performance of any online algorithm relative to that of an optimal offline algorithm. Specifically, we give three fundamental conditions for finite, 2, and \(\frac{3}{2}\) competitive ratios in terms of the environment parameters. We then design and analyze three classes of online algorithms that have provably finite competitiveness under varying environmental parameter regimes. Finally, we give a numerical visualization of these regimes to show the comparative strengths and weaknesses of each algorithm. 

\end{abstract}

\section{Introduction}

This paper considers an online perimeter defense problem in a tree graph environment in which a single defending vehicle attempts to intercept intruders before they enter a region of the graph known as the perimeter. Tree graphs commonly arise in situations such as rural or underdeveloped road systems \cite{basnet1999heuristics}. The scenario studied here might arise when it is necessary for autonomous vehicles to intercept or track targets in such a location, e.g. drones intercepting incoming vehicles to warn of upcoming road hazards.  In this paper we confine ourselves to a class of tree graphs known as full trees. Full trees allow for both a convenient parameterization scheme and contain the most possible intruder entrances and the largest perimeter for any tree that satisfies those parameters. This perimeter defense problem is online in the sense that the defending vehicle has complete information of an intruder only when it is present in the environment and is not aware of when or where new intruders will enter the environment.

Perimeter defense problems have previously been studied in a variety of environments under varying sets of assumptions (cf.~the survey~\cite{shishika2020review}). These studies tend to focus on strategies for either a small number of intruders~\cite{von2020multiple} or intruders released as per some stochastic process \cite{macharet2020adaptive,smith2009dynamic,bullo2011dynamic}. These problems have also been extensively studied in the context of reach-avoid and pursuit-evasion games~\cite{das2022guarding,pourghorban2022target}. However, these works typically do not consider a worst-case scenario in which large numbers of intruders are deployed strategically to overwhelm the defender(s). Recognizing this, \cite{bajaj2021competitive} and \cite{bajaj2022competitive} utilized the competitive analysis technique to design algorithms for worst case inputs in linear and conical environments. This technique \cite{sleator1985amortized,gutierrez2006whack} characterizes the performance of an online algorithm by comparing its performance to that of an optimal offline algorithm on the same input. This value is known as the \emph{competitive ratio} of an online algorithm. 

Another related body of work is graph clearing in which a team of searching agents try to capture a mobile intruder who is aware of the searchers' locations. Strategies for locating both mobile and static intruders in such a scenario were explored for graphs by \cite{borra2015continuous}, and in partial information settings by~\cite{kalyanam2016pursuit,sundaram2017pursuit}. Determining the number of searchers necessary to guarantee the capture of the intruder is known to be NP-complete for general graphs, but can be found in linear time for trees \cite{megiddo1988complexity}. Similarly, determining an optimal search strategy is NP-hard on graphs but can be done in polynomial time for trees \cite{kolling2009pursuit}. This differs from the perimeter defense problem considered here in that we have only a single defender (searcher) and each intruder is locked to a fixed course.


The primary contribution of this paper is to generalize the scope of the algorithmic strategies and fundamental limits derived for linear \cite{bajaj2021competitive} and conical \cite{bajaj2022competitive} environments to a new class of tree-based environments. Specifically, we consider defense in a rooted full tree with edges of unit length, total depth \(d\in\mathbb{N}\), and all vertices of distance less than \(d\) from the root having \(\delta\in\mathbb{N}\) children. The perimeter is then defined as all vertices at a distance less than or equal to \(\rho\in \{1,\dots, d-1\}\) from the root vertex. Intruders enter the environment from the tree's leaves and move at a fixed speed of \(v< 1\) towards the nearest vertex in the perimeter. A single defending vehicle moving with a maximum speed of 1 patrols the environment to capture as many intruders as possible before they reach the perimeter. The specifics of the environment are detailed in Section~\ref{section:problem-formulation}. 

Our specific contributions are as follows, we establish necessary conditions for the intruder velocity \(v\) in terms of the tree parameters \(d\), \(\rho\), and \(\delta\) for any online algorithm to be \(c\)-competitive for a finite \(c\), 2-competitive, or \(\frac{3}{2}\)-competitive (Section~\ref{section:fundamental-limits}). We then design three online algorithms that have provable competitiveness under various parameter regimes. Specifically, we give an algorithm that is 1-competitive, an algorithm whose competitiveness is a function of \(\delta\) and \(\rho\) for the environment, and a class of algorithms whose competitiveness varies based on a sweeping depth parameter (Section~\ref{section:algorithms}). Finally, we provide a numerical visualization of the parameter regimes in which the fundamental limits and algorithm guarantees apply (Section~\ref{section:numerical-viz}).

\section{Problem Formulation}
\label{section:problem-formulation}

In this section, we formally introduce the models of the environment, defending vehicle and intruder vehicles. We begin by formally defining a full tree.
Let \(T = (V,E,\mathcal{W})\) be a weighted tree graph rooted at vertex \(r\in V(T)\) with weight function, \(\mathcal{W}\), assigning unit weight to each edge in \(E\).  Further, for vertices \(v_{i}, v_{j} \in V(T)\), we say that \(v_{i}\) is the \emph{child} of \(v_{j}\) if \((v_{i}, v_{j}) \in E(T)\) and \(\textit{dist}(v_{i}, r) > \textit{dist}(v_{j}, r)\), where the distance \(\textit{dist}\, : V \times V \to \mathbb{R}_{\geq 0}\) denotes the length of the shortest path between any two vertices measured over the weights of the path's edges. \(T\) is a \emph{full tree} with depth \(d\in \mathbb{N}_{\geq 2}\) and branching factor \(\delta \in \mathbb{N}_{\geq 2}\) if: 
\begin{enumerate}
    \item \(\forall v \in V(T),\)  \(dist(v,r) \leq d\)
    \item Each vertex \(v \in V(T)\) such that \(dist(v,r) < d\) has exactly \(\delta\) child vertices, and
    \item Each vertex \(v \in V(T)\) such that \(dist(v,r) = d\) has no child vertices.
\end{enumerate}

An environment \(\enve(d, \delta, \rho)\) is the full tree with depth \(d\) and branching factor \(\delta\) rooted at some vertex \(r\). Within \(\enve(d, \delta, \rho)\), is the set of \emph{perimeter vertices} \(P(\enve)\). \(P(\enve)\) is defined as the set of all vertices \(v\) such that \(dist(v,r) = \rho\). These vertices can be thought of as the boundary between the perimeter region of the environment and the region the intruders move through to reach the perimeter. As \(\enve(d, \delta, \rho)\) is a full tree, the number of vertices \(|P(\enve)| = \delta^{\rho}\). Next, we define the set of intruder entrances for an environment, \(L(\enve)\), as the set of all vertices \(v\) such that \(dist(v,r) = d\). This gives us that \(|L(\enve)| = \delta^d\). A visualization of the environment \({\enve(d=2,\rho=1, \delta=2)}\) is given in Figure~\ref{fig:sample-environment}.


\paragraph{Intruders}
An intruder is a mobile agent that enters the environment from any of the vertices in \(L(\enve)\). After entering the environment, an intruder moves at a fixed speed \( v \in(0, 1)\) along the shortest path to the nearest vertex in \(P(\enve)\). Intruders do not change their speed or deviate from this shortest path.

\paragraph{Defending Vehicle}
As in \cite{bajaj2021competitive}, the defense is a single vehicle with motion modeled as a first order integrator. The defending vehicle always begins at the root vertex \(r\) and can freely traverse the environment via its edges, moving at a maximum speed of unity. We also assume that the defender is aware of the intruder velocity \(v\). 

\paragraph{Capture and Loss}
An intruder is said to be captured if its location coincides with that of the defender before the intruder reaches a perimeter vertex. An intruder is said to be lost if it reaches a perimeter vertex without being captured. We give ties in capture or loss to the defender, which is to say that an intruder may be captured exactly on a perimeter vertex without being lost. 


The edges of the environment graph do not merely represent connectivity between vertices, but constitute continuous one dimensional spaces. Thus, at any time instant, a vehicle (defender or intruder) may be located either at a vertex or at some location on an edge. Consider an edge \(e\) and let \(v_i\) and \(v_j\) be its endpoints such that \(dist(v_{i}, r) < dist(v_{j}, r)\). We map the location of a vehicle on edge \(e\) to the interval \([0,1]\) such that 0 corresponds to \(v_i\) and 1 corresponds to \(v_j\). Thus, the location of a vehicle at time instant \(t\) is given by a tuple \((e, x)\), where \(e\) is an edge and \(x\in [0,1]\). We then define the set of intruder locations, \(\mathcal{Q}(t)\), as the location tuples of all intruders in the environment at time instant \(t\). 

Now consider a tuple of the form \((t, N(t), \mathcal{Z})\) where \(t\) is a time instant, \(N(t)\) is the number of intruders released at time \(t\), and \(\mathcal{Z}\) maps the \(N(t)\) intruders to the intruder entrances where they will be released.
An input instance \(\mathcal{I}\) is the set of such tuples for all \(0\leq t \leq F\), where \(F\) is some final time instant. Note that given \(\mathcal{I}\) and the intruder velocity \(v\), it is possible to derive \(\mathcal{Q}(t)\) for any \(t\leq F\). We can now formally define online and offline algorithms for the defending vehicle.

\begin{figure}[h]
    \centering
    \vspace{-10pt}
    \includegraphics[width=0.5\linewidth]{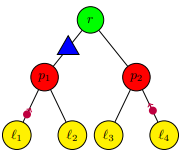}
    \caption{An environment with depth \(d=2\), perimeter depth \(\rho=1\), and branching factor \(\delta=2\). The small purple dots represent the intruders trying to reach perimeter vertices \(p_1\) and \(p_2\) after arriving from intruder entrances \(\ell_1\) and \(\ell_4\). The defending vehicle is represented by the blue triangle.}
    \label{fig:sample-environment}
    \vspace{-10pt}
\end{figure}

\vspace{-10pt}
\begin{definition}[Online Algorithm]
An \emph{online algorithm} 
\(\mathcal{A}\) is one which assigns the defender a speed in the interval \([0,1]\) and a direction along the shortest path between its current location and some vertex \(g\) as a function of the current locations of intruders in the environment and the intruder velocity.    
\end{definition}

\vspace{-18pt} 
\begin{definition}[Offline Algorithm]
An \emph{offline algorithm} is an algorithm \(\mathcal{O}\) that computes the defender's speed and direction as a function of the entire input instance \(\mathcal{I}\) and the intruder velocity \(v\). Such an algorithm is aware a priori of when and where all intruders will be released, i.e., a non-causal algorithm. We say an offline algorithm is optimal if captures the maximum number of intruders possible from \(\mathcal{I}\).

\end{definition}
\vspace{-15pt} 

\begin{definition}[Competitive Ratio, \cite{bajaj2021competitive, bajaj2022competitive}]
    Given an environment \(\enve(d,\delta,\rho)\), an input instance \(\mathcal{I}\), an intruder velocity \(v\), and an online algorithm \(\mathcal{A}\), let \(\mathcal{A}(\mathcal{I})\) be the number of intruders from \(\mathcal{I}\) captured by a defender using \(\mathcal{A}\). Now, let \(\mathcal{O}\) be an optimal offline algorithm that maximizes the number, \(\mathcal{O}(\mathcal{I})\), of captured intruders from \(\mathcal{I}\). Then, the \emph{competitive ratio of \(\mathcal{A}\) on input instance \(\mathcal{I}\)} is then \({c_{\mathcal{A}}(\mathcal{I}) = \frac{\mathcal{O}(\mathcal{I})}{\mathcal{A}(\mathcal{I})}}\). The \emph{competitive ratio of \(\mathcal{A}\) for an environment \(\enve\)} is \({c_{\mathcal{A}}(\enve) = \sup_{\mathcal{I}} c_{\mathcal{A}}(\mathcal{I})}\). Finally, \emph{the competitive ratio of an environment \(\enve\)} is \(c(\enve) = \inf_{\mathcal{A}} c_{\mathcal{A}}(\enve) \). We say that an online algorithm \(\mathcal{B}\) is \emph{\(c\)-competitive for an environment \(\enve(d,\delta,\rho)\)} if \(c_{\mathcal{B}}(\enve) \leq c\), for some constant \(c \geq 1\).
\end{definition}

\vspace{-8pt}
Note that it is preferable for an online algorithm to have a low competitive ratio, as this corresponds to it having similar performance to what is optimal. For instance, a 1-competitive algorithm matches the performance of an optimal offline algorithm for all inputs. Under these definitions, determining the competitive ratio of online algorithms is commonly done by considering inputs that offer some clear advantage to an optimal offline that is not available to an online defender.




\medskip

\emph{\bf Problem Statement:} The goal is to derive fundamental limits on the competitiveness of any online algorithm and to design online algorithms for the defender on the full tree environment and characterize their competitiveness.

\section{Fundamental Limits}
\label{section:fundamental-limits}

We give three fundamental limits on the competitiveness of any online algorithm in terms of the environment's parameters. First, we adapt the results of \cite{bajaj2021competitive} to the tree environment by finding optimal linear environments embedded within the full trees to show bounds on finite and 2-competitiveness. We then give a limit that is fully unique to the tree environment that characterizes parameter regimes in which online algorithms cannot perform better than \(\frac{3}{2}\)-competitive. We first define some relevant concepts.

\vspace{-7pt}
\begin{definition}[\textit{Descendant Vertices}]
For \(v_{i}, v_{j} \in V(\enve)\), if there exists a path from the root \(r\) to some \(\ell \in L(\enve)\) that includes both \(v_{i}\) and \(v_{j}\), we say that \(v_{i}\) is the \emph{descendant} of \(v_{j}\) if \(\textit{dist}(v_{i}, r) > \textit{dist}(v_{j}, r)\).
\end{definition}
\vspace{-13pt}
\begin{definition}[\textbf{Branch of a Tree Graph}]
    For a vertex \(v \in V(\enve)\), the \emph{branch of \(\enve\)} rooted at \(v\), denoted \(B(v)\), is the induced subgraph of \(\enve\) consisting of \(v\) and all descendants of \(v\). \(B(v)\) is a full tree with branching factor \(\delta\) and depth \((d - dist(r,v))\).
\end{definition}
\vspace{-13pt}
\begin{definition}[\textbf{Branch Entrances}]
    For a branch \(B(v)\), the \emph{branch entrances} are \(ent(B(v)) = L(\enve) \cap V(B(v))\).
\end{definition}
\vspace{-10pt}
\begin{theorem}
\label{thm:finit-com}
    For any environment \(\mathcal{E}(d, \delta, \rho)\), if the intruder velocity \(v > \frac{d-\rho}{2\rho}\), then there does not exist a \(c\)-competitive algorithm for any finite, positive number \(c\).
\end{theorem}
This first result, whose proof is presented in the Appendix, characterizes parameter regimes for which no online algorithm can be finitely competitive. 
The next result characterizes parameter regimes in which an online algorithm cannot do better than 2-competitive. Notably, the proof technique for this result is unique to the tree, as it relies on having multiple intruder entrances under each perimeter vertex.

\begin{theorem}
    \label{thm:two-comp}
    For an environment \(\mathcal{E}(d, \delta, \rho)\), if \(v \geq \frac{d-\rho}{d+\rho}\), then \(c(\mathcal{E}) \geq 2\).  
\end{theorem}

\begin{proof}
    We first select two perimeter vertices, \(p_i\) and \(p_j\) such that \(dist(p_{i},p_{j}) = 2\rho\). We then show that for two input instances of two intruders each, we can guarantee that only a single intruder can be captured by an online algorithm, while an optimal offline can always capture both. Let \(\ell_{i} \in ent(B(p_{i}))\) and \(\ell_{j} \in ent(B(p_{j}))\).

     We first show the result when \(v=\frac{d-\rho}{d+\rho}\). Consider an input where a single intruder is released at each of \(\ell_i\) and \(\ell_j\) at time \(d\). The only possible method to capture both intruders in this input is for the defender to be located at either of \(\ell_i\) or \(\ell_j\) at time \(d\) (capturing one intruder right away) and then moving immediately to either \(p_j\) (if it first captured the intruder at \(\ell_i\)) or \(p_i\) (if it first captured the intruder at \(\ell_j\)). Since it takes \(d\) times units to reach either of \(\ell_i\) or \(\ell_j\) from \(r\), any algorithm that does not immediately begin moving to either \(\ell_i\) or \(\ell_j\) at the start of the scenario can be at best 2-competitive. However, since an online algorithm cannot know which of the intruder entrances will be \(\ell_i\) and which will be \(\ell_j\), there always exists an alternate input instance when the location it arrives at is incorrect. Meanwhile, an optimal offline algorithm is aware of the correct vertex to move to and can always capture both intruders. 

     For the case when \(v>\frac{d-\rho}{d+\rho}\), we can use a similar input where an intruder is released at \(\ell_i\) at time \(d\) and a second intruder is released at \(\ell_j\) at time \(2d+\rho-\frac{d-\rho}{v}\). The only method to capture this input is to arrive at \(\ell_i\) at time \(d\) (capturing the first intruder) and then immediately moving to \(p_j\) to capture the second intruder. Once again, an online algorithm's lack of knowledge about which intruder entrance will be \(\ell_i\) means that there always exists an input where any given algorithm loses an intruder. 

     In summary, we have described two classes of inputs such that, without prior knowledge of where the intruders will be deployed, no single online algorithm can capture both intruders. Since there exist offline algorithms that capture all intruders in these inputs, no online algorithm can be better than 2-competitive, and the result follows. 
\end{proof}

The next result, proven in the Appendix, shows parameter regimes in which no online algorithm can do better than \(\frac{3}{2}\)-competitive.

\vspace{-10pt}
\begin{theorem}

    \label{thm:not-one-comp}
    For an environment \(\mathcal{E}(d,\delta, \rho)\) such that \(\delta \geq 3\), if the following two conditions hold:
    \begin{equation}\frac{d-\rho}{d + 3\rho} \leq v < \frac{d-\rho}{d+\rho}, \text{ and }\end{equation}
    \begin{equation}\label{eq:eps} d + \rho + 2(d-\rho)\frac{1-v}{1+v} - \frac{2\epsilon v}{1+v}> \frac{d-\rho}{v},\end{equation}
        for $\epsilon = d+3\rho - \frac{d-\rho}{v}$. Then, \(c(\mathcal{E}) \geq \frac{3}{2}\). 
\end{theorem}

\vspace{-6pt}
Although this result shows only \(\frac{3}{2}\)-competitiveness, it bounds the competitiveness of online algorithms for values of \(v\) that have not previously been explored, using inputs that are unique to the tree environment. 



\section{Algorithms}
\label{section:algorithms}
We now present online algorithms for the defending vehicle and analyze their competitiveness.
\subsection{Sweeping Algorithm}

In this algorithm, the defending vehicle continuously traverses every edge in the tree in a fixed order. Finding such a walk is analogous to finding the shortest length closed walk with starting vertex \(r\) that incorporates every edge of the tree. This walk can be found easily by following the path induced by a \emph{depth-first-search} on the tree beginning at \(r\). We will assume that the Sweeping algorithm follows the path of a depth-first-search that chooses the left-most unvisited vertex when it must make a decision.

\vspace{-10pt}
\begin{lemma}
    \label{lem:sweep-length}
    The length of one iteration of a Sweeping algorithm on an environment with depth \(d\) and branching factor \(\delta\) is \(2((\frac{\delta^{d+1} - 1}{\delta-1}) -1)\). 
\end{lemma}

This result follows from the fact that each edge in the tree must be traversed exactly twice, and  yields the following.

\vspace{-6pt}

\begin{theorem}
    A Sweeping algorithm is 1-competitive if 
    \begin{equation}
        \label{equ:sweep-bound}
        v \leq \frac{d-\rho}{2((\frac{\delta^{d+1} - 1}{\delta-1}) -1) - (d-\rho)}
    \end{equation}
    Otherwise, it is not \(c\)-competitive for any finite \(c\).
\end{theorem}
\vspace{-10pt}
\begin{remark}
    This result characterizes the fact that all intruders in an input can certainly be captured, even in an online setting, if the intruder velocity is sufficiently small. 
\end{remark}


\subsection{Stay At Perimeter Algorithm}
We next give an adaptation of an algorithm first presented in \cite{bajaj2022competitive} that allows for significantly more permissive parameter regimes at the cost of an exponential competitive ratio, Stay at Perimeter (SaP). This algorithm breaks up time into epochs of duration in the interval \([2\rho , 4\rho]\). During each epoch, the defending vehicle either waits at one of the perimeter vertices of the environment, capturing any intruders that are headed toward that location or travels to a new perimeter vertex and waits there. 


The SaP algorithm is detailed in Algorithm \ref{alg:snp} and is illustrated in Figure~\ref{fig:SaP}. We provide a brief overview of the terminology necessary to understand the algorithm. Let \(p_{1}, \dots , p_{\delta^{\rho}}\) be the perimeter vertices in the environment, and let \(P_{1}, \dots P_{\delta^{\rho}}\) be the subtrees given by \(B(p_{1}), \dots , B(p_{\delta^{\rho}})\). For a subtree \(P_k\), the regions \(S^{1}_k\), \(S^{2}_k\), and \(S^{3}_k\) are defined as all locations in \(P_k\) whose distance from \(p_k\) are in the intervals \([0,2\rho v]\), \([2\rho v, 4\rho v ]\), and \([4\rho v, 6\rho v ]\) respectively. We denote the number of intruders present in each of these regions as \(|S^{h}_k|\) for \(h\in \{1,2,3\}\). Suppose, the defender is located at some perimeter vertex \(p_j\) at the beginning of an epoch. The algorithm finds the value \(\eta_{j}^{k}\) for each subtree \(P_k\) by performing the following calculation:

   
      \[ \eta_{j}^{k} := 
    \begin{cases}
        |S_{k}^{2}| + |S_{k}^{3}|, & \text{if } k \neq j \\
        |S_{k}^{1}| + |S_{k}^{2}| + |S_{k}^{3}|, & \text{if } k=j.
    \end{cases}\]

That is, if \(P_k\) is the subtree the defender waited at during the previous epoch (\(k=j\)), the value of \(\eta_{j}^{k}\) is equal to the total number of intruders in all three regions of the subtree \(P_k\). Otherwise (\(k\neq j\)), \(\eta_{j}^{k}\) is equal to the number of intruders in just the second and third regions of \(P_k\). This accounts for the (worst case) travel time the defender will incur by switching its location. After \(\eta_{j}^{k}\) is calculated for each subtree, the algorithm calculates vertex \(k^* = \argmax_{k\in\{1, \dots, \delta^\rho\}}\{\eta^{1}_{i}, \dots , \eta^{\delta^{\rho}}_{i}\}\) i.e., the vertex of the subtree which corresponds to the greatest number of intruders that would be captured in the next epoch. If \(k^{*}\) is the same vertex as \(p_j\) the defender does not switch locations. In the case that \(k^{*} \neq p_j\), the algorithm determines if \(|S_{k^{*}}^2| \geq |S_{j}^1|\). If so, then the defender moves to \(k^{*}\) and captures \(|S_{k^{*}}^2|\). Otherwise, the defender stays at \(p_j\).



\vspace{-4pt}

\begin{algorithm}[h]
    
    \caption{Stay at Perimeter}
    \begin{algorithmic}[1]
    \label{alg:snp}
    \STATE Defender is at vertex $r$ and waits until time \(2\rho\)
    \STATE $k^*$ = $\argmax_{k\in\{1, \dots, \delta^\rho\}}\{\eta^{1}_{i}, \dots , \eta^{\delta^{\rho}}_{i}\}$, set \(P_{i} = P_{k^{*}}\)
    \STATE Move to vertex $p_{k^*}$ and wait until time $6\rho$
    \FOR{\emph{each} epoch}
       \STATE $k^*$ = $\argmax_{k\in\{1, \dots, \delta^\rho\}}\{\eta^{1}_{i}, \dots , \eta^{\delta^{\rho}}_{i}\}$ 
       \STATE \(P_{n} = P_{k^{*}}\)
       \IF{\(P_{n} \neq P_{i}\) \AND \(|S_{n}^{2}| \geq |S_{i}^{1}|\)}
            \STATE Move to vertex \(p_{k^*}\)
            \STATE Wait for \(2\rho \) time units, capturing \(|S_{n}^{2}|\)
            \STATE \(P_{i} = P_{n}\)
       \ELSE
            \STATE Remain at current vertex
            \STATE Wait for \(2\rho \) time units, capturing \(|S_{i}^{1}|\)
       \ENDIF
    \ENDFOR    
    
    \end{algorithmic}
    
\end{algorithm}
\vspace{-12pt}

\begin{figure}
    \centering
    \includegraphics[width=0.48\textwidth]{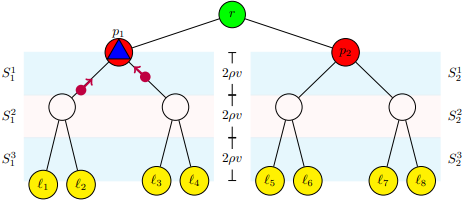}
    \caption{Breakdown of an environment as perceived by the Stay at Perimeter (SaP) algorithm. Here SaP considers the two branches rooted at \(p_1\) and \(p_2\). The regions of the subtrees considered in each epoch are highlighted. As long as the defender is located at \(p_1\) all intruders from the left subtree will be captured.}
    \label{fig:SaP}
    \vspace{-10pt}
\end{figure}

\begin{theorem}
    For any environment that satisfies  \(v \leq \frac{d-\rho}{6\rho}\), Stay at Perimeter is \(\frac{3 \cdot \delta^\rho - 1}{2}\)-competitive.
\end{theorem}

\begin{proof} To prove this result, we will establish a one to one correspondence between the Stay at Perimeter (SaP) algorithm and the Stay Near Perimeter (SNP) algorithm from \cite{bajaj2022competitive} with the following modifications to its parameters. The \(n_s\) resting points of SNP become the \(\delta^{\rho}\) perimeter vertices of the full tree environment for SaP.  The effect of positioning the defender at a perimeter vertex in the full tree setting is equivalent to positioning a defender at a resting point in the conical setting as it prevents the loss of any intruder from the branch rooted at that vertex so long as the defender is there. The need for a capture radius around the defender is eliminated in the tree environment as the defender needs to only occupy a single point-like location to block off that section of the perimeter from losses rather than a region like in the cone. The sectors considered by SNP then become the branches rooted at each perimeter vertex for SaP. The distance \(D\) seen in the description of SNP is equal to \(2\rho\) in the full tree. This is because if \(p_i\) and \(p_j\) are the perimeter vertices in a full tree environment \(2 \leq dist(p_{i},p_{j}) \leq 2\rho\). For this reason, the time intervals of duration \(Dv\) for the conical environment are equivalent to the \(2\rho v\) time intervals described above. Finally, as the full tree environments considered here are of some whole number depth \(d\) rather than the radius of 1 used in \cite{bajaj2022competitive}, our constraint for the existence of the intervals becomes \(\frac{d-\rho}{v} \leq 6\rho\).

    With these modifications, the results of Lemmas IV.5 and IV.6 from \cite{bajaj2022competitive} also become applicable to SaP. This is because these results do not rely on any specific property of the conical environment. Instead, they only depend on the distance \(D\), the number of sectors/resting points \(n_s\), and the comparisons made by SNP. As SaP has an equivalent notion of \(D\) and \(N_s\) and makes the same decisions as SNP based on these notions, we can apply these results to our analysis. This gives us that: every two consecutive intervals captured by SaP account for \(3(\delta^\rho -1)\) lost intervals, and every interval lost by SaP is accounted for by some captured interval. Given that these results hold, we have that Stay at Perimeter is \(\frac{3\cdot \delta^\rho -1}{2}\)-competitive when \(v \leq \frac{d-\rho}{6\rho}\).
\end{proof}

\medskip

The parameter regimes under which the SaP algorithm can perform do not depend on the branching factor \(\delta\) of the environment, instead it only requires that the previously mentioned regions are well formed. However, this is traded for a competitive ratio dependent on the number of perimeter vertices in the environment which does depend exponentially on \(\delta\). In the next algorithm, we will examine a strategy that seeks to strike a balance between sweeping the entire tree and only waiting at the perimeter vertices.

\subsection{Compare and Subtree Sweep Algorithm}
We now give a new algorithm that sweeps only a portion of the environment, allowing for a more permissive parameter regime. The Compare and Subtree Sweep algorithm (CaSS) takes a single additional parameter, an integer sweeping depth \(1 \leq s \leq \rho\), which determines both the competitiveness of the algorithm as well as the parameter regimes under which that competitiveness can be achieved. The algorithm breaks up time into epochs; during each epoch one of the \(\delta^s\) subtrees rooted at a vertex of distance \(s\) from the root is swept using the previously described sweeping method.

Compare and Subtree Sweep is defined in Algorithm \ref{alg:CaSS} and is summarized as follows. Let \(m_{1} , \dots m_{\delta^s}\) be the set of vertices of distance \(s\) from the root vertex \(r\). Now consider the subtrees rooted at each of these vertices. Each of these subtrees contain an equal, positive number of perimeter vertices. For each vertex \(m_k\), the capture region \(M_k\) is defined as all locations in the subtree rooted at \(m_k\) whose distance from a perimeter vertex is at least \(\frac{1}{2}(d-\rho)\). We denote the number of intruders present in \(M_k\) as \(|M_{k}|\). At the beginning of each epoch, the defending vehicle is located at vertex \(r\), where it identifies the capture region with the greatest number of intruders \(M_{*}\) (breaking any ties by choosing the left most region). The defender then moves to the root vertex of the subtree containing \(M_*\), denoted \(m_*\), and carries out the Sweeping algorithm on the subtree starting at vertex \(m_{*}\). The vehicle then returns to \(r\), and the next epoch begins. A possible path the defender might follow in an epoch is illustrated in Figure~\ref{fig:CaSS}.

\begin{figure}[h]
    \centering
    \includegraphics[width=0.45\textwidth]{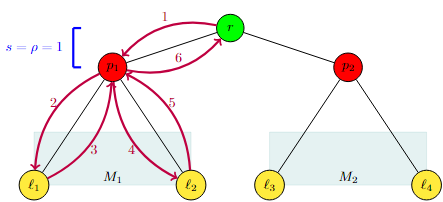}
    \caption{The Compare and Subtree Sweep algorithm will follow the path in purple when \(s=1\) and \(|M_{1}| > |M_{2}|\) (shown as the shaded regions). As \(\rho=1\) in this environment, there is only a single valid value for \(s\).} \label{fig:CaSS}
    \vspace{-12pt}
\end{figure}
\vspace{-6pt}
\begin{lemma}
    \label{lemma:CaSS-Conditions}
    For any environment such that,
    \[4v\left ( s+\frac{\delta^{d-s+1}}{\delta - 1} -1 \right ) \leq d-\rho,\]
    \begin{itemize}
        \item Every intruder that lies in \(M_{*}\) at the beginning of an epoch \(k\) is captured by Compare and Subtree Sweep with sweeping depth \(s\), and  
        \item Any intruder that enters the environment during the course of an epoch \(k\) is either captured during epoch \(k\) or is located in a capture region during epoch \(k+1\).
    \end{itemize}
\end{lemma}

\begin{proof}
    We begin by considering the length of time taken by the defending vehicle to complete a single epoch. The vehicle must first travel from \(r\) to the root of the subtree containing \(M_*\) taking time \(s\). The defender must then complete a single iteration of the Sweeping algorithm on the subtree. Since the subtree in question is of depth \(d-s\) and the defending vehicle moves with unit speed, Lemma \ref{lem:sweep-length} gives us that this takes time \(2(\frac{\delta^{d-s+1}-1}{\delta - 1}-1)\). Finally, the vehicle must return to the root taking another \(s\) time units. Thus a single epoch takes a total of \(2(s + \frac{\delta^{d-s+1}-1}{\delta - 1}-1)\) time units. This means that an intruder in \(M_*\) can move at most distance \(2v(s + \frac{\delta^{d-s+1}-1}{\delta - 1}-1)\) during an epoch. Even assuming that the intruder is as close to a perimeter vertex as possible, while being within \(M_*\) and is moving at the maximum velocity permitted by the constraint above, it can only travel at most distance \(\frac{1}{2}(d-\rho)\). However, during this time the entire subtree containing \(M_*\) has been traversed by the defender and the defender has returned to the root, implying the intruder in question has been captured. The first result follows.

    For the second result, we first consider the case that an intruder enters the environment in a different subtree than \(M_*\) for epoch \(k\). Suppose that it enters in the subtree containing capture region \(M_j\). We first note that the region \(M_j\) includes the intruder entrances for the subtree and thus the intruder begins in the capture region \(M_j\). By the same analysis as above, the furthest an intruder can travel in an epoch is \(\frac{1}{2}(d-\rho)\). Thus, even given the maximum travel time within epoch \(k\) (which would occur when the intruder enters just as epoch \(k\) begins), the intruder still lies within the capture region as epoch \(k+1\) begins and the result follows. We now consider the case that an intruder enters the environment in the same subtree as \(M_{*}\) for epoch \(k\). There exist time intervals within epoch \(k\) such that intruders arriving in those intervals will be captured in epoch \(k\) and thus will not be located in a capture region during epoch \(k+1\). The most obvious of these intervals is the first \(s\) time units of epoch \(k\), where the defending vehicle is moving to vertex \(m_*\). Intruders entering during this interval will be within the subtree rooted at \(m_*\) at the beginning of the sweep and will be captured during it. The other intervals arise from the order in which the Sweeping algorithm visits intruder entrances in \(M_*\). First note that the Sweeping algorithm visits every intruder entrance in \(M_*\) exactly once from left to right. Intruders that enter during the course of epoch \(k\) at an intruder entrance that has not yet been visited during the sweep will be captured in epoch \(k\) as they are still in the path of the sweep. Meanwhile, intruders that enter at an intruder entrance that has been already been visited during epoch \(l\) will still be in the capture region during epoch \(k+1\) by the same token as in the previous case. This concludes the proof of the second result.
\end{proof}

\vspace{-8pt}
\begin{theorem}
\label{thm:CaSS-Comp}
Compare and Subtree Sweep with sweeping depth \(s\) is \(\delta^{s}\)-competitive in environments where 
\[4v\left ( s+\frac{\delta^{d-s+1}}{\delta - 1} -1 \right ) \leq d-\rho\]
\end{theorem}


\begin{proof}
    From Lemma \ref{lemma:CaSS-Conditions}, we ensure that Compare and Subtree Sweep captures \(\frac{1}{\delta^s}\) of all intruders entering the environment in every epoch. Thus, the result follows.
\end{proof}

\begin{algorithm}[h]
    
    \caption{Compare and Subtree Sweep}
    \begin{algorithmic}[1]
    \label{alg:CaSS}
    \renewcommand{\algorithmicrequire}{\textbf{Input:}}
    \STATE Defender is at vertex $r$ and waits for $2(s+\frac{\delta^{d-s+1}}{\delta - 1} -1)$ time units after the arrival of the first intruder
    \FOR{\textit{each} epoch}
    \STATE $M_*$ = $\argmax \{|M_{1}|, \dots , |M_{\delta^s}|\}$ 
    \STATE $m_*$ = root vertex of subtree containing $M_*$
    \STATE Move to $m_*$ and perform a Sweep on $M_*$
    \STATE Move to $r$
    \ENDFOR

    \end{algorithmic}
\end{algorithm}
\vspace{-12pt}

\begin{figure}[t]
    \centering
    \vspace{-15pt}
    \includegraphics[width=0.48\textwidth]{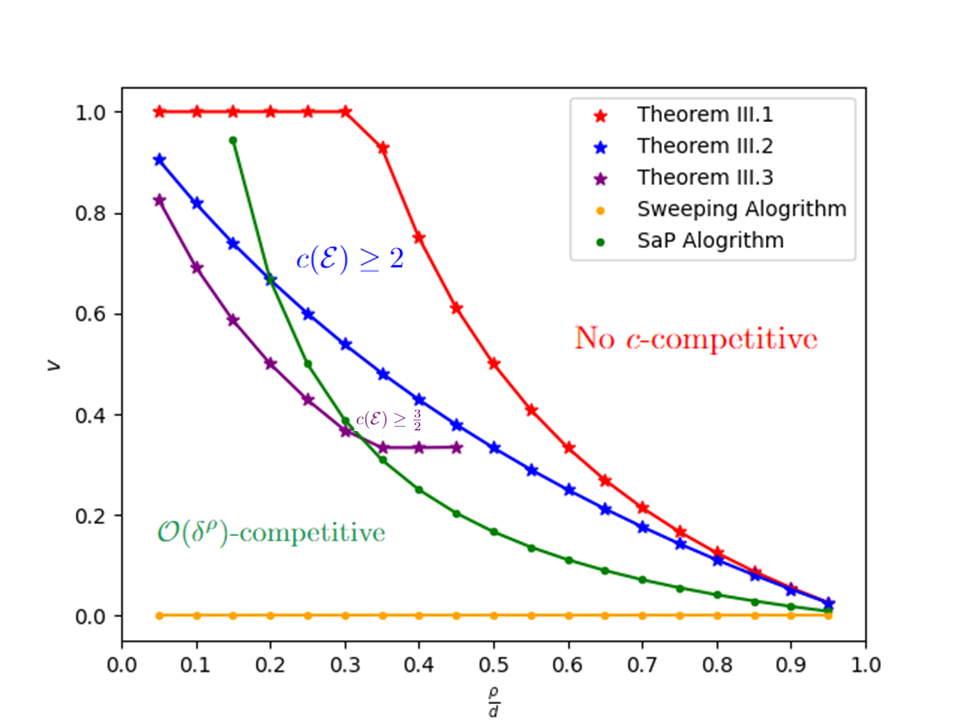}
    \caption{Parameter regimes for fundamental limits, Sweeping algorithm, and Stay at Perimeter Algorithm for varying values of \(\rho\) on a full tree with depth \(d=20\) and branching factor \(\delta=3\). The SaP algorithm is said to be \(\mathcal{O}(\delta^{\rho})\) competitive as its competitiveness is given by \(\frac{3 \cdot \delta^{\rho} -1}{2}\), where the $\mathcal{O}(\cdot)$ refers to the Landau notation.}
    \vspace{-18pt}
    \label{fig:nice-algorithms}
\end{figure}

\section{Numerical Visualizations}\label{section:numerical-viz}
We give a numerical visualization of the bounds derived for the full tree environment. Figure~\ref{fig:nice-algorithms} shows the \((\frac{\rho}{d} , v)\) parameter regimes for a fixed value of \(d=20\) and \(\delta = 3\) and a varying value of \(\rho\). Each point in the figure represents one of the possible integer values of \(\rho\) for the environment.

Values of \(v\) above the points corresponding to Theorem~\ref{thm:finit-com} and Theorem~\ref{thm:two-comp} correspond to velocities where there exist no \(c\)-competitive or no algorithm whose performance is better than 2-competitive, respectively. The space between the points corresponding to Theorem~\ref{thm:not-one-comp} and Theorem~\ref{thm:two-comp} give intruder velocities for which no algorithm can do better than \(\frac{3}{2}\)-competitive. Interestingly, this region only exists for environments where \(\frac{\rho}{d}<0.5\). This is due to the extra constraint on the environment parameters required by Theorem~\ref{thm:not-one-comp}.

It is unsurprising that the bound for the Sweeping algorithm is close to zero for all values of \(\rho\) as the bound is exponential with respect to the environment depth \(d\). While the bound for the SaP Algorithm, is significantly more permissive, it comes at the cost of exponential competitiveness with respect to \(\rho\). This may be an acceptable trade-off when \(\rho\) is small. However, a Sweeping strategy gives significantly better competitiveness for only a slightly more strict speed requirement as \(\rho\) approaches \(d\). 

Figure~\ref{fig:CaSS-Plot} shows the parameter regimes for which the Compare and Subtree Sweep algorithm is effective for several different sweeping depths. We see that for \(d=5\) and \(\delta=2\), CaSS  offers a larger effective area than the Sweeping algorithm for values of \(s\) that exceed 1. Thus, it is not advantageous to deploy CaSS with \(s=1\) for this combination of \(d\) and \(\delta\), as it offers a worse competitive ratio at a stricter velocity requirement. This is not always the case, however. For instance, setting \(\delta=3\) causes CaSS to always eclipse the parameter regime of the Sweeping algorithm. 

\begin{figure}[h]
    \centering
    \vspace{-10pt}
    \includegraphics[width=0.45\textwidth]{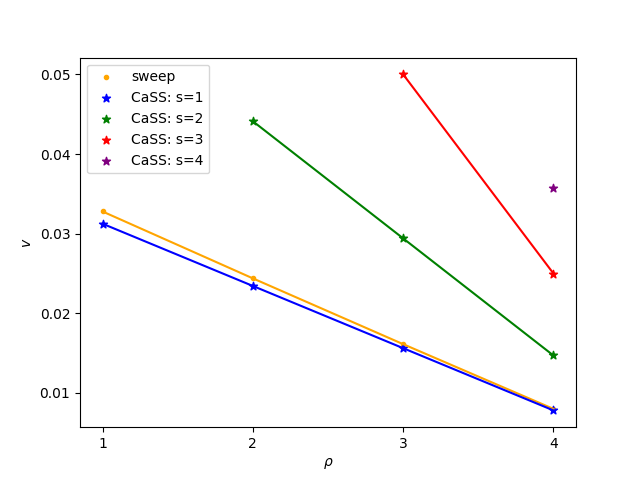}
    \vspace{-13pt}
    \caption{Parameter regimes for Compare and Subtree Sweep Algorithm and Sweeping Algorithm for all values of \(\rho\) on a full tree with depth \(d=5\) and branching factor \(\delta=2\).}
    \label{fig:CaSS-Plot}
    \vspace{-15pt}
\end{figure}

\section{Conclusions and Future Directions}
\label{section:conclusions}
This paper analyzed a scenario in which a single defending vehicle must defend a perimeter in a full tree environment from intruders that may enter at any time. We designed and analyzed three algorithms for the defending vehicle each with a provably finite competitive ratio. Specifically, the Sweeping algorithm is 1-competitive, i.e., it matches the performance of an optimal offline, but requires an exponential scaling constraint on the intruder's velocity. The Compare and Subtree Sweep algorithm offers a slight improvement on this constraint, and gives a range of competitiveness based on an externally chosen sweeping depth. Finally, the Stay at Perimeter algorithm is much more permissive in terms of its constraint on the intruder velocity. However, this is at the cost of an competitive ratio that scales exponentially with the perimeter depth \(\rho\), making it most effective when the perimeter size is small. We also derived three fundamental limits on the competitiveness of any online algorithm.

The work presented here suggests that there is a trade-off for online algorithms between exponential competitiveness and exponentially low intruder velocities. However, the exponential competitiveness results are based on the implicit assumption that an optimal offline algorithm can capture all intruders. To this end, future directions will include an improved worst-case analysis for the optimal offline. Additionally, we plan to expand the current analysis to trees that are not full, and create strategies for multiple defenders.
\vspace{-15pt}

\appendix

\subsection{Proof of Theorem~\ref{thm:finit-com}}
Inspired by the idea from \cite{bajaj2021competitive}, we will construct an input instance for \(\mathcal{E}\) consisting of two parts: a stream and a burst. We first select two perimeter vertices \(p_{i}, p_{j} \in P(\enve)\) such that \(dist(p_{i} , p_{j}) = 2\rho\). Since \(\enve\) is a full tree and \(\delta > 1\), there is always such a pair of vertices. Now select \(\ell_{i}\) and \(\ell_{j}\) arbitrarily from \(ent(B(p_{i}))\) and \(ent(B(p_{j}))\) respectively. Note that intruders arriving at \(\ell_{i}\) (resp.~\(\ell_{j}\)) will be lost when they arrive at \(p_{i}\) (resp.~\(p_{j}\)).

Let a burst of intruders, \(burst(\ell, n)\), be the simultaneous arrival of \(n\) intruders into the environment at intruder entrance \(\ell \in L\). Further, let a stream of intruders, \(stream(\ell, t)\), be the repeated arrival of a single intruder into the environment at intruder entrance \(\ell \in L\) with a delay of \(t\) time units between arrivals. The first part of the input is a burst, \(burst(\ell_j, c+1)\), that arrives at the earliest time that any online algorithm arrives at \(p_i\). The second part of the input consists of a stream, \(stream(\ell_{i}, 2d)\), beginning at time \(d\) and terminating as soon as the burst is released. 

Suppose the defender adopts an algorithm that never moves to \(p_{i}\). In this case, the stream will not terminate and all released intruders will be lost. Since an offline algorithm can move to \(p_{i}\) before any intruders in the stream are even released, it can capture all intruders in the stream. Thus, the result holds for this class of online algorithms.

We now consider online algorithms that do eventually move to \(p_{i}\) and show that the result still holds. Suppose the online defender arrives at \(p_{i}\) at time \(t\). At this time \(t\), the burst will be released at \(\ell_j\). Since \(v > \frac{d-\rho}{2\rho}\), the defender cannot reach \(p_j\) before the burst is lost. However, since an optimal offline algorithm is aware of the timing of the burst's release, it can always arrive at \(p_j\) in time to capture it. 

Since it is impossible for any online algorithm to move the vehicle so that it can capture any of the \(c+1\) intruders in the burst, we must consider how many intruders from the stream can be captured in order to compute the competitive ratio. If \(t < d\), then the stream never begins and the online defender captures no intruders. Here, an optimal offline can simply move to \(p_j\) by time \(t+ \frac{d-\rho}{v}\), which it can always do as \(t \geq \rho\). Thus the competitive ratio is infinite for this scenario.  If \(t=d\), then a single intruder from the stream has been released when the burst is released. An online algorithm cannot capture the burst but can capture the single stream intruder by just waiting at \(p_i\). Meanwhile, an optimal offline defender can once again ignore the single stream intruder and capture the burst. As the offline algorithm captures \(c+1\) intruders and the online captures only a single intruder, this class of online algorithms is \(c+1\)-competitive. Finally, when \(t>d\), we have at least the same competitiveness result as when \(t=d\). As the delay between stream intruder releases is \(2d\) there is only ever a single stream intruder present in the environment at any given time. Thus the online defender can only ever capture a single intruder from the stream. Again the offline defender can always guarantee the capture of the burst guaranteeing \(c+1\)-competitiveness. Indeed, if \(t\) is sufficiently larger than \(d\), then the offline defender can capture some of the intruders from the stream before moving to capture the burst giving an even higher competitive ratio.
\vspace{-4pt}
\subsection{Proof of Theorem~\ref{thm:not-one-comp}}
To show this result, we describe a class of inputs that consist of 3 intruders deployed according to a specific schedule. First, we select 3 distinct perimeter vertices \(p_{1}\), \(p_{2}\), and \(p_3\) such that the distance between every pair \(p_i\), \(p_j\) is \(2\rho\). We then select \(\ell_1\), \(\ell_2\), and \(\ell_3\) from \(ent(B(p_{1}))\), \(ent(B(p_{2}))\), and \(ent(B(p_{3}))\) respectively. 

We first consider the case when \(v = \frac{d-\rho}{d+3\rho}\). Consider the following input. The intruder released at \(\ell_1\) is released at time \(d+3\rho\), the intruder released at \(\ell_2\) is released at time \(d+\rho\), and the intruder released at \(\ell_3\) is released at time \(d+3\rho\). We will refer to these intruders as \(A\), \(B\), and \(C\) respectively. As none of \(A\), \(B\), or \(C\) will be lost at the same perimeter vertex, there is no method to capture any of these intruders simultaneously. Therefore, in any strategy that captures all intruders there must be an intruder that is captured first, second, and last. To begin, we will show that there exist two solutions that capture all intruders with capture orders \(A, B, C\) and \(C, B, A\) respectively. We will then show that these are the only such solutions that capture \(A\) or \(C\) first. 

A possible capture solution with capture order \(A,B,C\) is as follows. At time \(d+3\rho\) the defending vehicle arrives at \(\ell_1\), capturing \(A\) instantaneously as it enters the environment. The defending vehicle then moves to capture \(B\) at \(p_2\). As \(p_2\) is distance \(d+\rho\) away from \(\ell_1\), the defending vehicle arrives at \(p_2\) at time \(2d+5\rho\), just in time to capture \(B\). Finally, the defending vehicle moves to \(p_3\) to capture \(C\). As \(p_3\) is at a distance of \(2\rho\) from \(p_2\), the defending vehicle has just enough time to do this. 
Notice that the captures of \(B\) and \(C\) occur at the latest time possible, with any deviation from the described route resulting in the loss of at least one of the intruders. This implies that any delay in capturing \(A\) would also result in the loss of an intruder. Therefore, any algorithm that hopes to capture all intruders and capture \(A\) first must follow exactly the path described above, meaning that there is only a single capture solution with capture order \(A,B,C\). Also, note that \(A,C,B\) is not a valid capture solution as even just the travel time for the defender from \(\ell_1\) (where it captures \(A\)) to \(p_3\) (where it must reach to capture \(C\)) causes the loss of \(B\). Therefore, the only solution for capturing all three intruders in the input that begins with capturing \(A\) is the one described above.  As \(A\) and \(C\) are released simultaneously, there exists a symmetric capture solution where the defending vehicle arrives at \(\ell_3\) at time \(d+3\rho\) capturing \(C\), then captures \(B\) at \(p_2\), and finally moves to \(p_1\) to capture \(A\). Reasoning as before, it follows that there is only a single method for capturing all intruders in this input that captures \(C\) first.

We now show that there is no method to capture all three intruders that begins by capturing \(B\). To do this, we will show that even if the defender captures \(B\) as early as possible (giving it the most time to capture \(A\) and \(C\)) and moves optimally to capture \(A\) and \(C\) (even if they have not yet been released) it still cannot capture all three intruders. The earliest point in time that \(B\) can be captured is just as it has been released at time \(d+\rho\). At this time, neither of \(A\) or \(C\) have been released (nor will they be for another \(2\rho\) time units). However, let us assume that the algorithm for the defender causes it to begin moving to towards \(p_1\) or \(p_3\) which would minimize the distance it must travel to capture \(A\) or \(C\) respectively. As \(A\) and \(C\) are released simultaneously (and thus lost simultaneously), we will assume without loss of generality that the defending vehicle moves towards \(\ell_1\) to capture \(A\). Now the soonest that the defending vehicle can reach \(p_1\) is at time \(2d+2\rho\). As \(d > \rho\), \(2d+2\rho > d+3\rho\), meaning that at the moment the defending vehicle arrives at \(p_1\), \(A\) and \(C\) have been in the environment for \(d-\rho\) time units. Thus, \(A\) and \(C\) are at a distance of \((d-\rho)(1-v)\) from \(p_1\) and \(p_3\), respectively. Starting from \(p_1\), it will take the defending vehicle a minimum of \(2(d-\rho)\frac{1-v}{1+v}\) time units to capture \(A\) and return to \(p_1\). It will then require an additional \(2\rho\) time units to reach \(p_3\), which it must do before \(C\) is lost at time \(d+3\rho+\frac{d-\rho}{v}\). This means that in order to capture all intruders, we must have the following: \((d+\rho) + (d+\rho) + 2(d-\rho)\frac{1-v}{1+v} + 2\rho \leq d+3\rho+\frac{d-\rho}{v}\).
However, this relation is a direct contraction to condition~\eqref{eq:eps}, as \(\epsilon=0\) when \(v=\frac{d-\rho}{d+3\rho}\). Therefore, all three intruders can not be captured in strategies where \(B\) is captured immediately. It can be easily shown that not capturing \(B\) immediately does not offer any improvement.



In summary, when \(v=\frac{d-\rho}{d+3\rho}\), the only method to capture all intruders in the described input is for the defending vehicle to either arrive at \(p_1\) or \(p_3\) exactly at time \(d+3\rho\). However, algorithms that arrive at \(p_1\) by this time, can only capture at most two intruders in an input that releases intruders at \(p_2\) and \(p_3\) at time \(d+3\rho\) and at \(p_1\) at time \(d+\rho\). Similarly, algorithms that arrive at \(p_3\) by this time, can only capture at most two intruders in an input that releases intruders at \(p_1\) and \(p_2\) at time \(d+3\rho\) and at \(p_3\) at time \(d+\rho\). Therefore, without prior (offline) knowledge of where  intruder \(A\) or \(C\) will be released and which vertices are \(\ell_1\), \(\ell_2\), and \(\ell_3\) no algorithm can capture all three intruders. Therefore, \(c(\mathcal{E}) \geq \frac{3}{2}\) when \(v=\frac{d-\rho}{d+3\rho}\).

The case of when \(v > \frac{d-\rho}{d+3\rho}\) follows a similar argument. However the input changes to: \(A\) is released at time \(d+3\rho\) at \(\ell_1\), \(B\) is released at time \(d+\rho + \epsilon\) at \(\ell_2\), where \(\epsilon = d+3\rho - \frac{d-\rho}{v}\), and \(C\) is released at time \(d+3\rho + \epsilon\) for the same \(\epsilon\) at \(\ell_3\). For this input, the only valid capture order is \(A\), \(B\), \(C\) and the result follows by a simlar argument to before.

\addtolength{\textheight}{-12cm}   






\bibliographystyle{ieeetr}
\bibliography{refs}

\end{document}